\documentclass[conference]{IEEEtran}
\IEEEoverridecommandlockouts
\usepackage{cite}
\usepackage{amsmath,amssymb,amsfonts}
\usepackage{algorithmic}
\usepackage{graphicx}
\usepackage{textcomp}
\usepackage{xcolor}
\def\BibTeX{{\rm B\kern-.05em{\sc i\kern-.025em b}\kern-.08em
    T\kern-.1667em\lower.7ex\hbox{E}\kern-.125emX}}

\usepackage{amsthm}
\newtheorem{theorem}{Theorem} 
\newtheorem{lemma}{Lemma} %
\usepackage{threeparttable}
\usepackage{multirow}
\usepackage{algorithm,algorithmic}
\usepackage{booktabs} 
\begin{document}

\title{SERE: A Stabilized Element-Wise Method for Downlink Rate Estimation in Clustered Cell-Free Networks\\
}

\author{
\IEEEauthorblockN{Panpan Niu\IEEEauthorrefmark{1}, Han Hao\IEEEauthorrefmark{1}, Hao Wu\IEEEauthorrefmark{1}, and Junyuan Wang\IEEEauthorrefmark{2}
\thanks{The first two authors contributed equally to this work.}
\thanks{$^{\dagger}$Corresponding author.}
}

\IEEEauthorblockA{\IEEEauthorrefmark{1}\textit{Department of Mathematical Sciences, Tsinghua University, Beijing, China}}

\IEEEauthorblockA{\IEEEauthorrefmark{2}\textit{ College of Electronic and Information Engineering,
Tongji University, Shanghai, China}
\IEEEauthorblockA{Email: \{npp21, haoh23\}@mails.tsinghua.edu.cn, hwu@tsinghua.edu.cn, junyuanwang@tongji.edu.cn}
}
}

\maketitle

\begin{abstract}
Clustered cell-free networks have emerged as a promising architecture for sixth generation ultra-dense wireless communication systems by enabling local cooperation among base stations while controlling system complexity.
For resource allocation and performance optimization of such networks, accurate and efficient estimation of the ergodic achievable downlink rate is a fundamental prerequisite.
Existing rate estimation approaches mainly rely on computationally prohibitive Monte Carlo simulations or adopt random matrix theory-based methods, which have been well-developed for conventional cellular and cell-free networks.
However, existing RMT-based methods have not addressed the unique inter-subnetwork interference in clustered cell-free networks, and therefore lack an efficient solution for accurate downlink rate estimation under both regularized zero-forcing and zero-forcing precoding.
In this paper, we propose a stabilized element-wise rate estimation method for downlink rate estimation in clustered cell-free networks. 
We establish the diagonal element-wise convergence of resolvent matrices, which enables the derivation of deterministic equivalents for inter-subnetwork interference and the downlink ergodic rate. We further introduce a stabilized variable transformation to address the numerical instability when the regularization parameter is very small, hereby enabling a unified formulation applicable to both regularized zero-forcing and zero-forcing precoding.
Simulation results show that the proposed method achieves a relative error below 6\% while significantly reducing computational complexity compared with the Monte Carlo simulation.
\end{abstract}

\begin{IEEEkeywords}
clustered cell-free networks, downlink rate, random matrix theory, stabilized element-wise rate estimation
\end{IEEEkeywords}

\vspace{-0.7em}
\section{Introduction}
\vspace{-0.5em}

The growing demand for wireless communications and the increasing user density have revealed critical performance limitations in traditional cellular networks, particularly in terms of providing reliable service to edge users~\cite{gesbert2010multi}.
This issue has driven the exploration of alternative network architectures.
One promising solution is clustered cell-free networks~\cite{wang2022rate,wang2023clustered, zhou2024energy}.
This architecture introduces local cooperation within a decentralized framework: it dynamically organizes base stations (BSs) and users that strongly interfere with each other into the same subnetwork and performs joint transmission in each subnetwork.
This approach reduces system complexity while maintaining high spectral efficiency and user fairness.

When designing and optimizing clustered cell-free networks, the ergodic achievable rate serves as a fundamental performance metric that requires frequent and efficient estimation, especially for resource management problems like power allocation and regularization parameter selection~\cite{wang2022rate}.
Nevertheless, the ergodic achievable rate involves statistical expectation with respect to random small-scale fading channels, making it intractable to derive an explicit closed-form expression directly.
To this end, Monte Carlo simulation~\cite{kalos2009monte} is widely used in practice to approximate the expectation by averaging over numerous independent channel realizations.
Monte Carlo simulation thus suffers from extremely high computational complexity due to the need for large-scale matrix inversion per channel realization, a problem that is unavoidable under both zero-forcing (ZF)~\cite{caire2003zf} and regularized ZF (RZF)~\cite{peel2005rzf} precoding.

To overcome this limitation, random matrix theory (RMT)-based methods~\cite{hachem2007deterministic, hachem2008clt} have emerged as a well-established alternative.
These RMT-based methods have been widely applied in deriving deterministic equivalents of the ergodic rate for various conventional wireless network architectures.
For the uplink, deterministic equivalents of the achievable rate with linear receivers have been derived for non-cooperative multi-cell massive MIMO systems~\cite{hoydis2013massive}.
For the downlink, the deterministic equivalent of the ergodic sum-rate was established for RZF precoding in correlated multiple-input single-output (MISO) broadcast channels with limited feedback~\cite{wagner2012large}.
These results were later extended to cooperative multi-cell downlink systems~\cite{zhang2013large} and cell-free massive MIMO architectures~\cite{wang2020performance, fu2026LPRZF}.
However, extending these RMT results to clustered cell-free networks under RZF and ZF precoding faces challenges. 
Inter-subnetwork interference brings unique cross-subnetwork random terms that do not appear in conventional cellular or fully cooperative cell-free systems, which hinders the derivation of corresponding deterministic equivalents.

Moreover, as the regularization parameter approaches zero, the resolvent matrix arising in the analysis of RZF precoding becomes ill-conditioned~\cite{horn2012matrix}, leading to severe numerical instability~\cite{wagner2012large,wang2020performance}.
Consequently, the deterministic equivalent for ZF precoding cannot be reliably derived from the RZF formulation by simply taking the limit $\alpha \to 0$.

In this paper, we propose the Stabilized Element-wise Rate Estimation (SERE) method for downlink rate estimation in clustered cell-free networks.
Targeting the inter-subnetwork interference, we establish the almost sure diagonal element-wise convergence of resolvent matrices, which complements existing trace-based results.
This allows us to decompose inter-subnetwork interference and derive its deterministic equivalent, as well as the downlink ergodic achievable rate.
We further propose a Stabilized Variable Transformation (SVT) to address the numerical instability in RZF rate estimation, which avoids resolvent matrix ill-conditioning and enables stable evaluation for both RZF and ZF precoding.
Numerical experiments demonstrate that the SERE method delivers fast estimates with a relative error below 6\%, validating its effectiveness.
 
\section{System Model}
Consider the downlink transmission of a clustered cell-free network consisting of $M$ disjoint subnetworks with multiple multi-antenna BSs and single-antenna users in each subnetwork.
Let $\mathcal{B}_m = \{b^m_1,\ldots,b^m_{L_m}\}$ and $\mathcal{U}_m = \{u^m_1,\ldots,u^m_{K_m}\}$ denote the sets of BSs and users associated with subnetwork $m$, respectively, with $L_m = |\mathcal{B}_m|$ and $K_m = |\mathcal{U}_m|$.
Here, $b^m_s$ and $u^m_k$ represent the $s$-th BS equipped with $N^m_{s}$ antennas and the $k$-th user in subnetwork $m$, respectively.
\vspace{-0.25em}
\subsection{Signal Model}
Let $u_k^m$ in subnetwork $m$ be the reference user. 
The received signal $y_k^m$ at user $u_k^m$ is given by
\vspace{-0.75em}
\begin{equation}\label{eq: signal decom}
\begin{aligned}
    y_k^m  = &\underbrace{\mathbf{g}^{m,m}_{k}\mathbf{x}^m_k}_{\text{desired signal}} +\underbrace{\sum_{j=1, j \neq k}^{K_m} \mathbf{g}^{m,m}_{k}\mathbf{x}^m_j}_{\text{intra-subnetwork interference}}
    \\
    & + \underbrace{\sum_{n=1, n \neq m}^{M} \sum_{j=1}^{K_n} \mathbf{g}^{m,n}_k\mathbf{x}^n_j}_{\text{inter-subnetwork interference}}
    + z_k^m,
\end{aligned}
\end{equation}
where $\mathbf{x}^m_{k} \in \mathbb{C}^{N_m\times 1}$ denotes the transmitted signal from the BS set $\mathcal{B}_m$ to user $u_k^m$, written as $\mathbf{x}^m_{k} =\mathbf{w}^m_{k}\cdot s_k^m$.
Here, $N_m=\sum_{s=1}^{L_m}N^{m}_{s}$ is the total number of transmit antennas in subnetwork $m$.
The symbol $s_k^m$ is the information-bearing signal for user $u_k^m$ satisfying $\mathbb{E}[s_k^ms_k^{m\dagger}] = 1$, and $\mathbf{w}^m_k\in \mathbb{C}^{N_m\times1}$ represents the corresponding precoding vector.
The channel gain vector from the BS set $\mathcal{B}_n$ to user $u_k^m$ is denoted by $\mathbf{g}^{m,n}_{k} \in \mathbb{C}^{1 \times N_n}$. Additive white Gaussian noise (AWGN) is denoted as $z_k^m$, which follows the distribution $\mathcal{CN}(0, N_0)$.

The channel gain vector is modeled as
$
    \mathbf{g}^{m,n}_{k}= \mathbf{l}^{m,n}_{k} \circ \mathbf{h}^{m,n}_{k},
$
where $\circ$ denotes the Hadamard product. 
The small-scale fading vector $\mathbf{h}^{m,n}_{k} \in \mathbb{C}^{1 \times N_n}$ consists of i.i.d. entries with complex Gaussian distribution $\mathcal{CN}(0,1)$, and $\mathbf{l}^{m,n}_{k}\in \mathbb{R}^{1 \times N_n}$ represents the large-scale fading vector between user $u_k^m$ and all transmit antennas in subnetwork $n$.
For each transmit antenna belonging to the $s$-th BS $b_s^n$, the corresponding large-scale fading coefficient is given by
\vspace{-0.75em}
\begin{equation*}
l_{k,s}^{m,n} = 
\begin{cases}
(d_{k,s}^{m,n})^{-1.75}, & d_{k,s}^{m,n} > d_1, \\
d_1^{-0.75} (d_{k,s}^{m,n})^{-1}, & d_0 < d_{k,s}^{m,n}  \leq d_1, \\
d_1^{-0.75} d_0^{-1}, & d_{k,s}^{m,n} \leq d_0,
\end{cases}
\end{equation*}
where $d_{k,s}^{m,n}$ denotes the Euclidean distance between user $u_k^m$ and BS $b_s^n$.
The parameters $d_0$ and $d_1$ correspond to the near-field and far-field thresholds, respectively.

Furthermore, the channel matrix between the users in $\mathcal{U}_m$ and the transmit antennas in subnetwork $n$ is denoted as 
$\mathbf{G}^{m,n}=[\mathbf{g}^{m,n\dagger}_{1},\cdots,\mathbf{g}^{m,n\dagger}_{K_m}]^\dagger \in \mathbb{C}^{K_m \times N_n}$.
For notational simplicity, $\mathbf{G}^{m} \triangleq \mathbf{G}^{m,m}$ and $\mathbf{g}^{m\dagger}_{k} \triangleq (\mathbf{g}^{m,m}_{k})^\dagger$ are defined.

\subsection{Achievable Rate}
Assuming that RZF precoding~\cite{peel2005rzf} is adopted in each subnetwork, the precoding matrix $\mathbf{W}^m=[\mathbf{w}^m_{1},\cdots,\mathbf{w}^m_{K_m}]\in\mathbb{C}^{N_m \times K_m}$ for the $m$-th subnetwork is given by
\begin{equation}\label{def: RZF matrix}
    \mathbf{W}^m =\xi_m\mathbf{F}^m = \xi_m \mathbf{G}^{m\dagger} \left( \mathbf{G}^{m} \mathbf{G}^{m\dagger} + N_m\alpha_m \mathbf{I}_{K_m}\right)^{-1},
\end{equation}
where $\alpha_m > 0$ is the regularization parameter.
The normalization factor $\xi_m$ satisfies the power constraint
$tr(\mathbf{W}^m\mathbf{W}^{m\dagger}) \leq P_m$, where $P_m$ is the available transmit power at the BSs in the $m$-th subnetwork.
Under the assumption that the $m$-th subnetwork fully utilizes its total transmit power, i.e., the power constraint is met with equality, the normalization factor $\xi_m$ can be derived as \cite{zhang2013large, wang2020performance}
\begin{equation}\label{eq:xi}
    \xi_m = \sqrt{\frac{P_m}{tr(\mathbf{F}^{m}\mathbf{F}^{m \dagger})}}, \ \forall m = 1, \cdots, M.
\end{equation}

Combining~\eqref{eq: signal decom}, \eqref{def: RZF matrix} and \eqref{eq:xi}, the received signal-to-interference plus noise ratio~(SINR) at user $u_k^m$ is expressed as
\begin{equation}\label{eq: SINR}
    \gamma_k =  \frac{D_k}{\frac{N_0}{\xi^2_m}+ I_{k}^{\text{intra}}+ I_{k}^{\text{inter}}},
\end{equation}
with\footnote{Here, $D_k$ represents the desired signal term of user $u_k^m$,
and $I_{k}^{\text{intra}}$ and $I_{k}^{\text{inter}}$ denote the intra-subnetwork and inter-subnetwork interference terms received at user $u_k^m$.}
\vspace{-0.5em}
\begin{align*}
    D_k&=\mathbf{g}^{m}_{k} \mathbf{f}^m_{k} \mathbf{f}^{m\dagger}_{k} \mathbf{g}^{m\dagger}_{k},\\
    I_{k}^{\text{intra}}&=\sum_{j=1, j \neq k}^{K_m} \ \mathbf{g}^{m}_{k}\mathbf{f}^{m}_{j}\mathbf{f}^{m\dagger}_{j}\mathbf{g}^{m\dagger}_{k},\\
    I_{k}^{\text{inter}}&=\sum_{n=1, n \neq m}^M \frac{\xi_n^2}{\xi_m^2}  \sum_{j=1}^{K_n} \mathbf{g}^{m,n}_{k} \mathbf{f}^n_{j} \mathbf{f}^{n\dagger}_{j} \mathbf{g}^{m,n\dagger}_{k},
\end{align*}
where $\mathbf{f}^{m}_{k}\in \mathbb{C}^{N_m\times 1}$ is the $k$-th column vector of $\mathbf{F}^m$. 

The ergodic achievable rate of user $u_k^m$ is then expressed as
\begin{equation}\label{eq: rate}
    R_k = \mathbb{E}_{\mathbf{H}}\left[\log_2\left(1 +\gamma_k\right)\right].
\end{equation}
The expectation is taken over the small-scale fading matrix $\mathbf{H}$, which covers all users and BSs in the network.

\section{Element-Wise Deterministic Equivalent Analysis}\label{sec:methed}
This section establishes the theoretical results for approximating the downlink ergodic achievable rate in clustered cell-free networks under RZF precoding.
We first establish the almost sure element-wise convergence of the resolvent matrices. Based on this convergence, we derive the deterministic equivalent of the downlink ergodic rate.

We consider the asymptotic regime in which the numbers of transmit antennas $N_m$ and users $K_m$ in each subnetwork approach to infinity while their ratios remain bounded, i.e., $\beta_m=N_m/K_m$
satisfies $0<\lim \inf \beta_m \leq \lim \sup \beta_m < \infty$ for all $m$.
Let $\beta = \sum_{m} N_m / \sum_{m} K_m$ denote the overall transmit-to-receive antenna ratio across the network.
The analysis begins with the resolvent matrices associated with the normalized Gram matrices in each subnetwork:
$\mathbf{Q}^{m}(z)\triangleq( N_m^{-1}\mathbf{G}^{m} \mathbf{G}^{m\dagger} -z\mathbf{I}_{K_m})^{-1}$ and $\tilde{\mathbf{Q}}^{m}(z)\triangleq(N_m^{-1} \mathbf{G}^{m\dagger} \mathbf{G}^{m} - z \mathbf{I}_{N_m})^{-1}$, for $z \in \mathbb{C} \setminus \mathbb{R}^+$.

\begin{lemma}\label{lemma:Q in hachem}
The following limits hold almost surely:\vspace{-0.5em}
\begin{align}
    \lim_{N_m \to \infty,\, N_m/K_m = \beta_m} [\mathbf{Q}^m(z)]_{ii} - \phi_{m,i}(z) &= 0, \label{as:diag Q}\\
    \lim_{N_m \to \infty,\, N_m/K_m = \beta_m} [\tilde{\mathbf{Q}}^m(z)]_{jj} - \psi_{m,j}(z) &= 0,\label{as:diag tilde Q}
\end{align}
where the vectors $\boldsymbol{\phi}_m(z) = [\phi_{m,1}(z), \dots, \phi_{m,K_m}(z)]^\dagger$ and $\boldsymbol{\psi}_m(z) = [\psi_{m,1}(z), \dots, \psi_{m,N_m}(z)]^\dagger$ are the solutions of the following $K_m + N_m$ equations:
\vspace{-0.5em}
\begin{align}
\phi_{m,i}(z)
&= \frac{1}{-z\bigl(1 + N_m^{-1}\sum_{j} \theta_{i,j}^m \psi_{m,j}(z)\bigr)}, \label{iteration phi}\\
\psi_{m,j}(z)
&= \frac{1}{-z\bigl(1 + N_m^{-1}\sum_{i} \theta_{i,j}^m \phi_{m,i}(z)\bigr)}, \label{iteration psi}
\end{align}
for $1\le i\le K_m$ and $1\le j\le N_m$, where $\theta_{k,l}^{m}=(l_{k,l}^{m,m})^2$.
\end{lemma}
\begin{IEEEproof}
    See Appendix~\ref{app: proof Q}.
\end{IEEEproof}
Lemma~\ref{lemma:Q in hachem} establishes convergence of the individual diagonal entries without requiring the $1/N_m$ normalization typically needed to maintain magnitude uniformity.
We further give the convergence results of squared resolvent matrices $\mathbf{S}^m(z) \triangleq (\mathbf{Q}^{m}(z))^2$ and $\tilde{\mathbf{S}}^m(z) \triangleq (\tilde{\mathbf{Q}}^{m}(z))^2$.
\begin{lemma}\label{lemma: Q2} \
The following limits hold almost surely:
\begin{align}
    \lim_{N_m \to \infty,\, N_m/K_m = \beta_m} [\mathbf{S}^m(z)]_{ii} - \lambda_{m,i}(z) &= 0, \label{as:diag S}\\
    \lim_{N_m \to \infty,\, N_m/K_m = \beta_m} [\tilde{\mathbf{S}}^m(z)]_{jj} - \mu_{m,j}(z) &= 0,\label{as:diag tilde S}
\end{align}
where the vectors $\boldsymbol{\lambda}_m(z) = [\lambda_{m,1}(z), \dots, \lambda_{m,K_m}(z)]^\dagger$ and $\boldsymbol{\mu}_m(z) = [\mu_{m,1}(z), \dots, \mu_{m,N_m}(z)]^\dagger$ are the solutions of the following $K_m + N_m$ equations:
\begin{align}
    \lambda_{m,i}(z) &= \phi_{m,i}^2(z) \Bigl[1 + N_m^{-1}\sum_j \theta_{i,j}^m (\psi_{m,j}(z)+z \mu_{m,j}(z)) \Bigr], \label{iteration lam}\\ 
    \mu_{m,j}(z) &= \psi_{m,j}^2(z)\Bigl[1 + N_m^{-1}\sum_i \theta_{i,j}^m (\phi_{m,i}(z) +z \lambda_{m,i}(z))  \Bigr],\label{iteration mu}
\end{align}
for $1\le i\le K_m$ and $1\le j\le N_m$.
\end{lemma}
\begin{IEEEproof}
See Appendix~\ref{app: proof Q2}.
\end{IEEEproof}

Utilizing these lemmas, the main result on rate approximation is established in the following theorem.
\begin{theorem}\label{them: DE of SINR}
There exists a deterministic equivalent $\bar{R}_k$ such that
\vspace{-0.5em}
\begin{equation*}
    \lim_{\substack{N_m \to \infty,\, N_m/K_m = \beta_m \\ m = 1, \ldots, M}} R_k - \bar{R}_k = 0.
\end{equation*}
\end{theorem}
\begin{IEEEproof}
The proof consists of three key stages.
First, we reformulate $D_k$, $I_k^{\text{intra}}$, and $I_k^{\text{inter}}$ into functions of the diagonal entries of resolvent matrices.
Next, leveraging large-dimensional RMT tools, we establish almost sure convergence limits for four terms $D_k$, $I_k^{\text{intra}}$, $I_k^{\text{inter}}$, and $\xi_m$.
Finally, these results are combined to yield the deterministic equivalent of the achievable rate.
\subsubsection{Reformulation} 
For the desired signal and the intra-subnetwork interference, we rewrite them as:\vspace{-0.5em}
\begin{align}
D_k&=|\mathbf{g}^{m}_{k}\mathbf{f}^{m}_{k}|^2=\left([ \mathbf{G}^{m} \mathbf{F}^m ]_{k,k} \right)^2,\label{step1:desired}\\
I_{k}^{\text{intra}} &=\sum_{j=1} ^{K_m}|\mathbf{g}^{m}_{k}\mathbf{f}^{m}_{j}|^2-|\mathbf{g}^{m}_{k}\mathbf{f}^{m}_{k}|^2\nonumber\\
&= \left[ \mathbf{G}^{m} \mathbf{F}^m \mathbf{F}^{m\dagger} \mathbf{G}^{m\dagger} \right]_{k,k} - \left(\left[ \mathbf{G}^{m} \mathbf{F}^m \right]_{k,k} \right)^2.\nonumber
\end{align}
For the inter-subnetwork interference $I_{k}^{\text{inter}}$, we first rewrite it in matrix form:
$I_{k}^{\text{inter}}
    =\sum_{n=1, n \neq m}^{M}\frac{\xi_n^2}{\xi_m^2}\mathbf{g}^{m,n}_k \mathbf{F}^n \mathbf{F}^{n\dagger} \mathbf{g}^{m,n\dagger}_k.
$
Note that the cross-subnetwork channel vector $\mathbf{g}_{k}^{m,n}$ is independent of the precoding matrix $\mathbf{F}^n$.
By leveraging the mutual independence of the elements in $\mathbf{g}_{k}^{m,n}$ and the trace lemma~\cite{Couillet_Debbah_2011}, we establish the following almost sure convergence:
\vspace{-0.5em}
\begin{equation*}\label{step1:inter}
 \lim_{\substack{N_m \to \infty,\, N_m/K_m = \beta_m \\ m = 1, \ldots, M}}
I_{k}^{\text{inter}} 
- \sum_{\substack{n=1 \\ n \neq m}}^{M} \frac{\xi_n^2}{\xi_m^2} \sum_{j=1}^{N_n} \theta_{k,j}^{m,n} \big[\mathbf{F}^n \mathbf{F}^{n\dagger}\big]_{j,j} 
= 0.
\end{equation*}
\vspace{-0.5em}

It is evident that the expressions for $D_k$, $I_k^{\text{intra}}$, $I_k^{\text{inter}}$ and $\xi_m$ are primarily determined by the diagonal entries of
$\mathbf{G}^{m} \mathbf{F}^m $, $\mathbf{G}^{m} \mathbf{F}^m \mathbf{F}^{m\dagger} \mathbf{G}^{m\dagger}$ and $\mathbf{F}^m \mathbf{F}^{m\dagger}, \ \forall m=1,\cdots,M$.
Therefore, the analysis shifts to characterizing the asymptotic behavior of the diagonal entries of these matrices.
By utilizing the definition of $\mathbf{F}^m$ in \eqref{def: RZF matrix}, we have\vspace{-0.5em}
\begin{equation}\label{eq:GF}
    \mathbf{G}^{m} \mathbf{F}^m
    =\mathbf{I}_{K_m}-\alpha_m (N_m^{-1}\mathbf{G}^{m} \mathbf{G}^{m\dagger} +\alpha_m \mathbf{I}_{K_m})^{-1}.
\end{equation}
Since $\mathbf{G}^m \mathbf{F}^m$ is Hermitian, it follows that
\begin{equation*}\label{eq:GFFG}
    \begin{aligned}
    \mathbf{G}^{m} \mathbf{F}^m \mathbf{F}^{m\dagger} \mathbf{G}^{m\dagger} 
    &=\mathbf{I}_{K_m}-2\alpha_m(N_m^{-1}\mathbf{G}^{m} \mathbf{G}^{m\dagger} +\alpha_m \mathbf{I}_{K_m})^{-1}\\
    &\quad \ +\alpha_m^2(N_m^{-1}\mathbf{G}^{m} \mathbf{G}^{m\dagger} +\alpha_m \mathbf{I}_{K_m})^{-2}.
\end{aligned}
\end{equation*}
Furthermore, we get
\begin{equation*}\label{eq:FF}
    \begin{aligned}
    \mathbf{F}^m \mathbf{F}^{m\dagger} &= \mathbf{G}^{m\dagger} ( \mathbf{G}^{m} \mathbf{G}^{m\dagger} + N_m\alpha_m \mathbf{I}_{K_m})^{-2}\mathbf{G}^{m} \\
    &=N_m^{-1}\left[(N_m^{-1} \mathbf{G}^{m\dagger} \mathbf{G}^{m} + \alpha_m \mathbf{I}_{N_m})^{-1}\right.\\
    &\left.\quad-\alpha_m(N_m^{-1} \mathbf{G}^{m\dagger} \mathbf{G}^{m} + \alpha_m \mathbf{I}_{N_m})^{-2}\right].
\end{aligned}
\end{equation*}

\subsubsection{Asymptotic Limits of $D_k$, $I_k^{\text{intra}}$, $I_k^{\text{inter}}$, and $\xi_m$} 
For the desired signal term $D_k$, substituting \eqref{eq:GF} into the \eqref{step1:desired} and applying Lemma~\ref{lemma:Q in hachem}, we obtain:
\begin{equation}\label{eq:as derivation}
\begin{aligned}
&\lim_{N_m \to \infty, N_m/K_m = \beta_m} D_k
-\left(1-\alpha_m\phi_{m,k}(-\alpha_m)\right)^2\\
&=\lim_{N_m \to \infty, N_m/K_m = \beta_m}\left[-2\alpha_m([\mathbf{Q}^m]_{kk}-\phi_{m,k}(-\alpha_m))\right.\\
&\left.\qquad\qquad\qquad\qquad+\alpha_m^2(([\mathbf{Q}^m]_{kk})^2-\phi_{m,k}^2(-\alpha_m))\right]\\
&=0.
\end{aligned}
\end{equation}
Thus, $D_k$ converges
almost surely to its deterministic equivalent $\bar{D}_k$
under the large-system regime, given by
\begin{equation}
\bar{D}_k=\left(1-\alpha_m\phi_{m,k}(-\alpha_m)\right)^2.
\label{eq:D with diagonal}
\end{equation}

Following the same approach as in \eqref{eq:as derivation}
and using Lemma~\ref{lemma: Q2}, we establish that
$I_k^{\text{intra}}$, $I_k^{\text{inter}}$, and $\xi_m$
converge almost surely to their deterministic equivalents
$\bar I_k^{\text{intra}}$, $\bar I_k^{\text{inter}}$, and
$\bar\xi_m$, respectively, in the large-system regime, given by
\begin{align}
    & \bar I_{k}^{\text{intra}} = \alpha^2_m\left(\lambda_{m,k}(-\alpha_m)-\phi^2_{m,k}(-\alpha_m)\right),\label{eq:I1 with diagonal}\\
    & \bar I_{k}^{\text{inter}} =\sum_{n=1,n\neq m}^M\frac{\bar{\xi}_n^2}{\bar{\xi}_m^2}N_n^{-1}\sum_{j=1}^{N_n}\theta_{k,j}^{m,n}\left(\psi_{n,j}(-\alpha_n)\right.\nonumber\\
    &\left.\qquad\qquad-\alpha_n\mu_{n,j}(-\alpha_n)\right)\label{eq:I2 with diagonal},\\
    & \bar{\xi}_m = \sqrt{\frac{P_m}{N_m^{-1}\sum_{j=1}^{N_m} \left(\psi_{m,j}(-\alpha_m)-\alpha_m\mu_{m,j}(-\alpha_m)\right)}}.\label{eq:xi^2 with diagonal}
\end{align}

\subsubsection{Deterministic Equivalent of ergodic achievable rate $R_k$}
Combining almost sure convergence results of $D_k$, $I_k^{\text{intra}}$, $I_k^{\text{inter}}$, and $\xi_m$ via the continuous mapping theorem~\cite{billingsley1995probability}, we establish the deterministic equivalent 
\begin{equation}\label{eq:bar rate}
\bar{R}_k=\log_2\left(1+\bar{\gamma}_k\right)=\log_2\left(1+\frac{\bar D_k}{\frac{N_0}{\bar \xi^2_m}+ \bar I_{k}^{\text{intra}}+ \bar I_{k}^{\text{inter}}}\right).
\end{equation}
Applying the dominated convergence theorem in \cite{billingsley1995probability}, we further show that the ergodic achievable rate converges to this deterministic equivalent:
$
    \lim_{\substack{N_m \to \infty,\, N_m/K_m = \beta_m \\ m = 1, \ldots, M}} R_k - \bar{R}_k =0.
$
This completes the proof.
\end{IEEEproof}
This theorem establishes the deterministic equivalent result for analyzing the performance of clustered cell-free network systems with RZF precoding in the large-system limit.
The derived deterministic equivalent not only enables computationally efficient system evaluation but also provides profound insights into the impact of key system parameters. 
Specifically, our method can facilitate the optimization of $\alpha_m$ and power allocation to maximize the ergodic sum achievable rate in clustered cell-free networks.

\section{Algorithm}\label{sec: alg}
\begin{algorithm}[t]
\renewcommand{\algorithmicrequire}{\textbf{Input:}}
\renewcommand{\algorithmicensure}{\textbf{Output:}}
\caption{Element-wise Subnetwork-Aware Method with Stabilized Variable Transformation}
\label{alg:Rate estimate under RZF} 
\begin{algorithmic}[1] 
    \REQUIRE Large-scale fading $\mathbf{l}^{m,n}_{k}$, regularization scalar $\alpha_m$
    \ENSURE Estimation of $R_k$
    \STATE 
    Iteration \eqref{iter phi2}-\eqref{iter psi2} for $T_{\max}$ iterations to compute $\boldsymbol{\phi}_m(-\alpha_m)$, $\boldsymbol{\hat \psi}_m(-\alpha_m)$ for each subnetwork, initializing $\boldsymbol{\phi}_m(-\alpha_m)$ with $\mathbf{1}_{K_m}$ \label{alg line phi}
    \STATE 
    Iteration \eqref{iter lambda2}-\eqref{iter mu2} for $T_{\max}$ iterations to compute $\boldsymbol{\lambda}_m(-\alpha_m)$, $\boldsymbol{\hat \mu}_m(-\alpha_m)$ for each subnetwork, initializing  $\boldsymbol{\lambda}_m(-\alpha_m)$ with $\mathbf{1}_{K_m}$\label{alg line mu}
    \STATE Calculate $\bar{\xi}_m$ with \eqref{eq:xi^2 with diagonal}
    \STATE Calculate $\bar{D}_k, \bar I_k^{\text{intra}}$, and $\bar I_k^{\text{inter}}$ by \eqref{eq:D with diagonal} to \eqref{eq:I2 with diagonal}
    \STATE Compute the approximation of $R_k$ by \eqref{eq:bar rate}
\end{algorithmic}
\end{algorithm}

The results of Theorem~\ref{them: DE of SINR} can be specialized to the ZF precoding~\cite{caire2003zf} case by taking the limit as $\alpha_m \to 0$, under the additional assumptions that \(\liminf_m \beta_m > 1\).
However, this limiting case requires careful treatment, since matrices $N_m^{-1}\mathbf{G}_m^\dagger \mathbf{G}_m+\alpha_m\mathbf{I}$ could become
rank-deficient, making the associated quantities
$\psi_{m,j}(-\alpha_m)$ and $\mu_{m,j}(-\alpha_m)$ numerically unstable as
$\alpha_m \to 0$.
To overcome this critical issue and establish a unified analytical method for both RZF and ZF precoding, we introduce the Stabilized Variable Transformation (SVT), which reformulates these intermediate variables to ensure numerical stability in the small-$\alpha$ regime.

The SVT defines the intermediate variables as
\begin{equation*}
\hat{\psi}_{m,j}(z) \gets z\psi_{m,j}(z), \ \
\hat{\mu}_{m,j}(z) \gets \psi_{m,j}(z)+z\mu_{m,j}(z).
\end{equation*}
These transformed variables remain $O(1)$ as $z \to 0$.
This boundedness can be proven via eigenvalue analysis of the matrices \( z\tilde{\mathbf{Q}}^m(z) \) and
\( \tilde{\mathbf{Q}}^m(z) + z\tilde{\mathbf{S}}^m(z) \), while the detailed derivation is omitted here for brevity.
By substituting SVT into the original fixed-point equations \eqref{iteration phi}–\eqref{iteration psi} and \eqref{iteration lam}–\eqref{iteration mu}, we obtain the numerically stable formulations
\begin{align}
\phi^{-1}_{m,i}(z) &= -z - N_m^{-1}\sum_j \theta_{i,j}^m \hat\psi_{m,j}(z), \label{iter phi2} \\
\hat\psi^{-1}_{m,j}(z) &= -( 1 + N_m^{-1}\sum_i \theta_{i,j}^m \phi_{m,i}(z) ),  \label{iter psi2}\\
\lambda_{m,i}(z) &= \phi_{m,i}^2(z) (1 + N_m^{-1}\sum_j \theta_{i,j}^m \hat{\mu}_{m,j}(z) ), \label{iter lambda2}\\
\hat\mu_{m,j}(z) &= \hat\psi_{m,j}^2(z)(N_m^{-1}\sum_i \theta_{i,j}^m \lambda_{m,i}(z) ),\label{iter mu2} \\
\forall 1\leq &i \leq K_m,\qquad  \forall 1 \leq j \leq N_m.\nonumber
\end{align}

\vspace{-0.5em}
The computational complexity of the proposed SERE method is $O\left(\sum_iK_i^2\right)$.
This represents a substantial reduction in complexity compared to Monte Carlo simulations, which reach a complexity of $O\left(\sum_iK_i^3\right)$ due to costly matrix inversions.
This reduction in complexity makes the SERE method well-suited for real-time optimization in future ultra-dense networks.
Moreover, $\boldsymbol{\phi}_m(-\alpha_m)$, $\boldsymbol{\hat\psi}_m(-\alpha_m)$, $\boldsymbol{\lambda}_m(-\alpha_m)$, and $ \boldsymbol{\hat\mu}_m(-\alpha_m)$ can be efficiently solved using fixed-point iterations~\cite{wagner2012large}. 
The complete algorithmic procedure is provided in Algorithm~\ref{alg:Rate estimate under RZF}.

\section{Performance Evaluation}
\vspace{-0.5em}
In this section, we demonstrate the accuracy and efficiency of the proposed SERE method, and validate the numerical stability of the SVT scheme for downlink rate estimation in clustered cell-free networks under both RZF and ZF precoding through simulations.
\begin{table}[t]
\centering
\caption{Experimental Parameter Settings}
\label{tab:experiment_parameters}
\begin{tabular}{cc}
\hline
Parameter & Value \\ \hline
Noise Power ($N_0$) & $1 \times 10^{-12}$ W \\
Power Constraint ($P_m$) & 1 W \\
Network Coverage Area ($D$) & 2000 m \\
Near field Threshold ($d_0$) & 10 m \\
Far field Threshold ($d_1$) & 50 m \\
Number of antennas per BS ($N^{m}_{s}$) & 1 \\
Regularization Parameter ($\alpha_m$)\cite{peel2005rzf} &  $ N_0(P_m\beta_m)^{-1}$\\ \hline
\end{tabular}\vspace{-0.5em}
\end{table}

Without loss of generality, we consider a scenario where both BSs and users are uniformly distributed within a square area.
The network is partitioned using the K-means clustering algorithm to group BSs and users into \( M \) non-overlapping subnetworks, as shown in the top row of Fig.~\ref{fig:rate under diff shapes and dists}.
We evaluate the average per-user rate of the central subnetwork, highlighted with a black circle.
To mitigate the impact of randomness in network construction, we computed the average downlink rate over 50 random clustered cell-free network realizations.
The basic experimental parameters are outlined in Table~\ref{tab:experiment_parameters}.
The Monte Carlo simulation averages the downlink rate over 100 random small-scale fading realizations.
The fixed-point iterations are run for $T_{\max}=50$, which is sufficient to ensure convergence in all tested scenarios. 
All experiments were conducted on a platform with two Intel Xeon E5-2699 v3 CPUs (36 cores) and 440G RAM.

\begin{figure}[t] 
    \centering
    \includegraphics[width=0.49\textwidth]{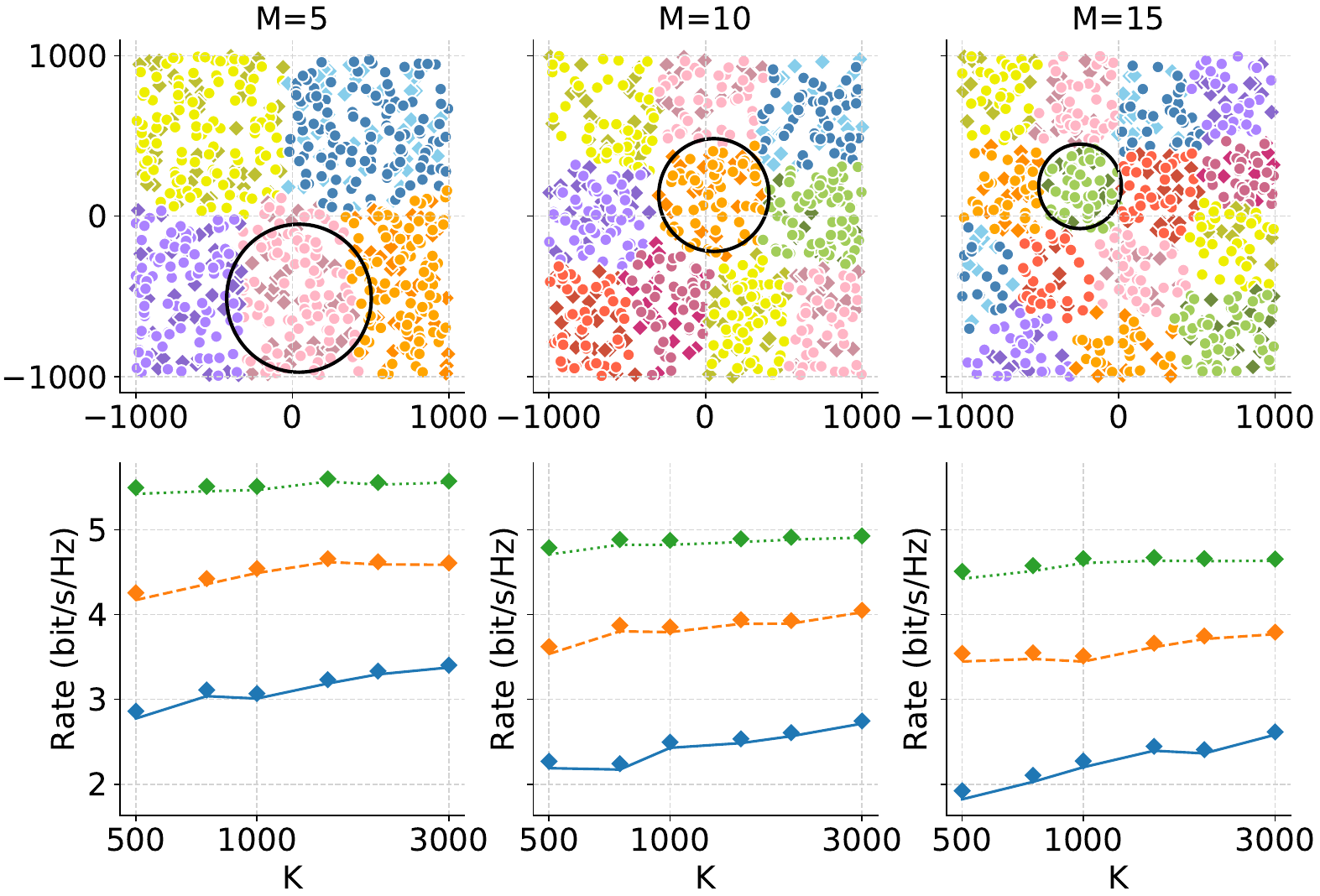}
    \vspace{-20pt}
    \caption{
    Top row: illustration of the clustered cell-free networks. Distinct colors denote different subnetworks, with diamonds representing users and circles corresponding to BSs.
    Bottom row: the downlink rate under RZF precoding versus the number of users $K$.
    Diamonds correspond to simulation results, and the lines represent our analytical approximations.
    There are three different ratios, which correspond to the green dotted line $\beta = 8$, the orange dashed line $\beta = 4$, and the blue solid line $\beta = 2$, respectively.
    }
    \label{fig:rate under diff shapes and dists}
\end{figure}
Fig.~\ref{fig:rate under diff shapes and dists} presents the visualization of clustered cell-free networks and the downlink rate performance under RZF precoding.
It shows that the average per-user rate obtained by the proposed SERE method aligns closely with the simulation results across different values of $M$ and $\beta$, indicating the method's high precision.
Moreover, the approximation accuracy improves as the number of BSs increases, which is consistent with the asymptotic convergence analysis presented in Theorem~\ref{them: DE of SINR}.
Remarkably, while our theoretical analysis is established under the large system limit, numerical simulations demonstrate that $\bar{R}_k$ maintains remarkable accuracy even at practical finite dimensions.

\begin{figure}[tbp] 
    \centering
    \includegraphics[width=0.49\textwidth]{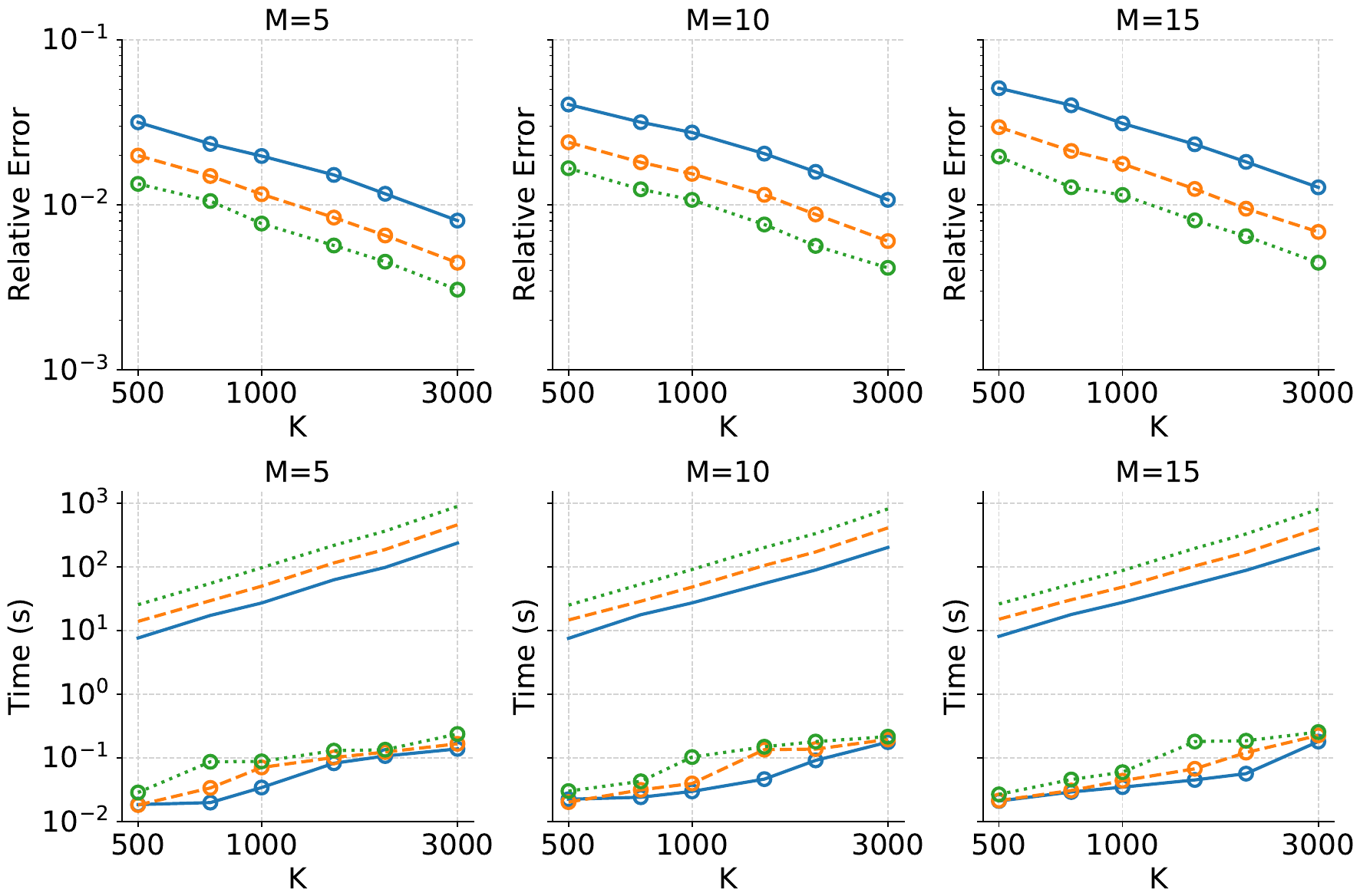}
    \vspace{-20pt}
    \caption{Performance of the SERE method versus the simulation baseline under RZF precoding. The top row shows the relative error of the average per-user rate estimation. The bottom row compares the computation time, where the lower and upper sets of curves correspond to the SERE method and the simulation, respectively. 
    The line styles denote various ratios: $\beta = 2$ (blue solid), $\beta = 4$ (orange dashed), and $\beta = 8$ (green dotted).}
    \label{fig:Relative error under rzf}
\end{figure}
Fig.~\ref{fig:Relative error under rzf} illustrates the relative error of our method as the $K$ increases. 
The results show a consistently decreasing trend, with relative errors below \(6\%\) across all parameter settings.
As $K$ increases, the relative error stabilizes around 
$1\%$, demonstrating strong reliability of the proposed method in large-scale scenarios.
In terms of computational efficiency, as illustrated in the second row of Fig.~\ref{fig:Relative error under rzf}, the proposed SERE method achieves an average speedup of 1279$\times$ compared to the baseline simulation.

\begin{figure}[tbp] 
    \centering
    \includegraphics[width=0.49\textwidth]{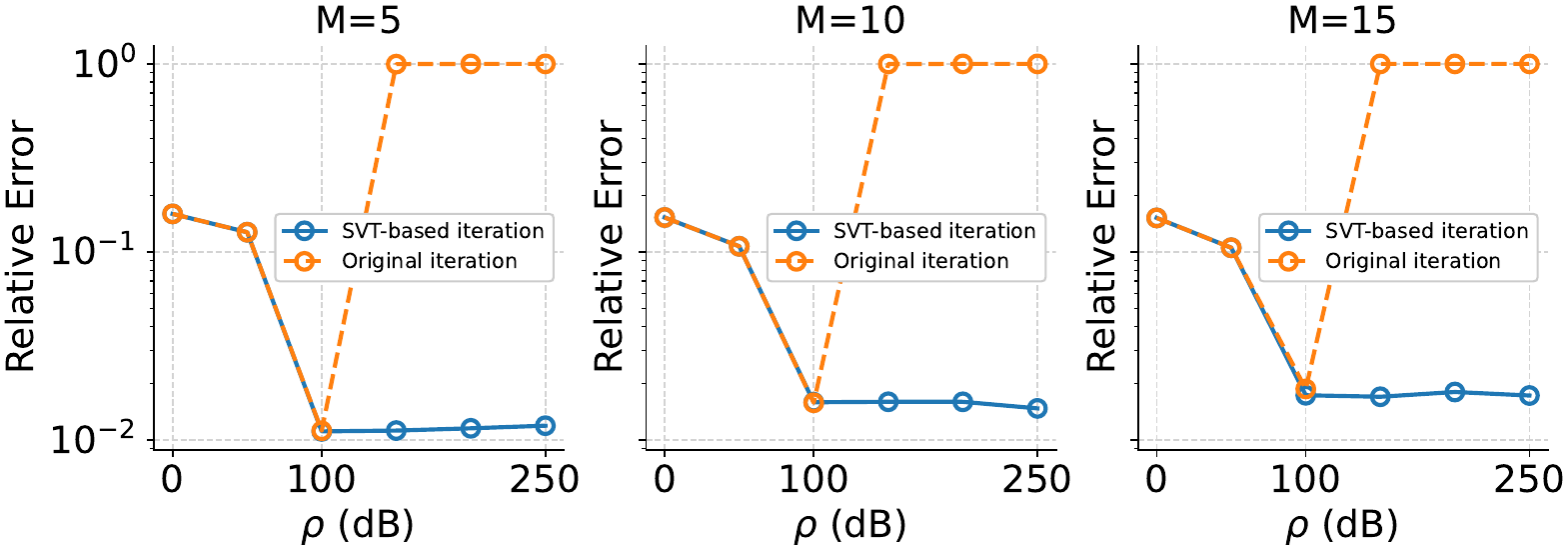}
    \vspace{-20pt}
    \caption{
    Comparison of the relative error of the average per-user rate with RZF precoding for the proposed SVT-based iteration and the original iteration, shown as a function of SNR.
    $K=1000$. $\beta=4$.}
    \label{fig:SVT_validation}
\end{figure}
\begin{figure}[t] 
    \centering
    \includegraphics[width=0.49\textwidth]{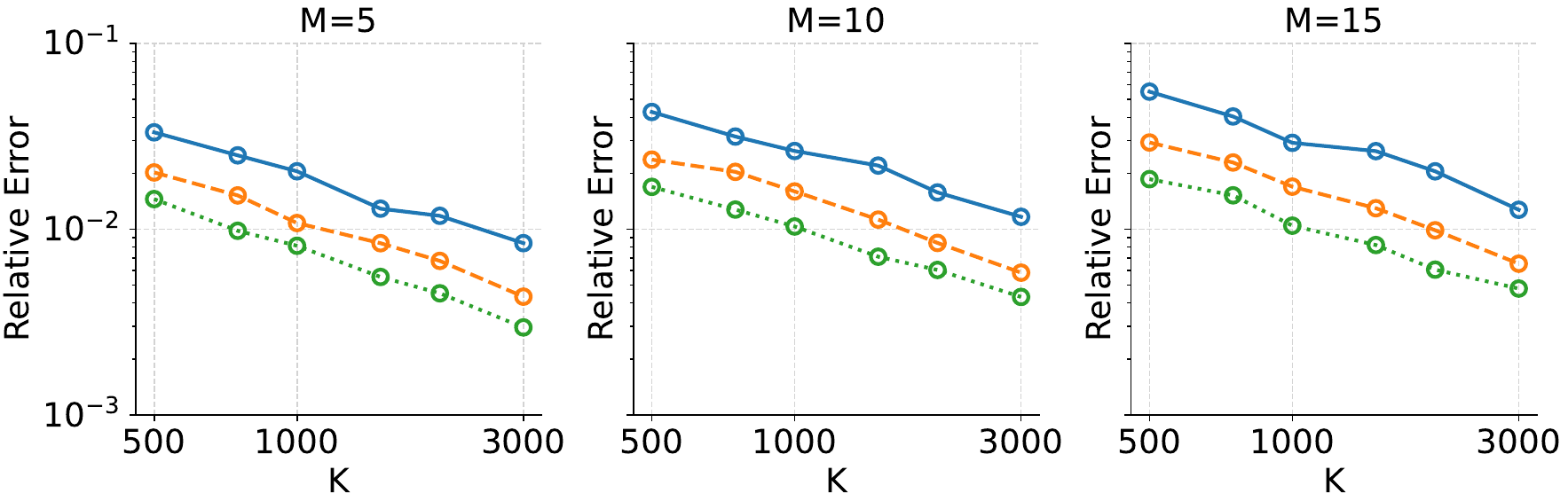}
    \vspace{-20pt}
    \caption{Relative error of the average per-user rate with our method under ZF precoding. The three ratios correspond to the green dotted line ($\beta = 8$), the orange dashed line ($\beta = 4$), and the blue solid line ($\beta = 2$), respectively.}
    \label{fig:SVT_ZF}
\end{figure}
To further verify the effectiveness of the SVT, we compare the rate estimation performance of the original fixed-point iterations defined by equations \eqref{iteration phi}–\eqref{iteration psi} and \eqref{iteration lam}–\eqref{iteration mu} against the proposed SVT-based formulation in \eqref{iter phi2}-\eqref{iter mu2}.
Fig.~\ref{fig:SVT_validation} plots the relative error versus the Signal-to-Noise Ratio~(SNR) $\rho$ for different numbers of subnetworks $M$, where the regularization parameter is set to $\alpha_m=(10^{\frac{\rho}{10}}\beta_m)^{-1}$.
In the high-SNR regime, which corresponds to extremely small $\alpha_m$, the original fixed-point iterations suffer from severe numerical instability, as discussed in Section~\ref{sec: alg}, causing the relative error to increase dramatically.
By contrast, the proposed SVT-based iterations maintain stable and accurate rate estimates across all $\rho$.
Fig.~\ref{fig:SVT_ZF} further demonstrates that the proposed SVT formulation remains well-defined and yields precise results in the ZF precoding case.

\vspace{-0.5em}
\section{Conclusion}
This paper proposed the Stabilized Element-wise Rate Estimation (SERE) method for efficient downlink rate estimation in clustered cell-free networks. 
The theoretical foundation was
established through the element-wise deterministic equivalents
of the resolvent matrices. 
Building upon this foundation, a
Stabilized Variable Transformation (SVT) was introduced to
ensure numerical stability and provide a unified
method that covers both RZF and ZF precoding.
Numerical results demonstrated that the proposed method achieves less than 6\% relative error while reducing computational complexity by over three orders of magnitude compared to Monte Carlo simulations.
\vspace{-1em}
\bibliographystyle{IEEEtran.bst}
\bibliography{reference.bib}
\appendices
\vspace{-1em}
\section{The Proof of Lemma~\ref{lemma:Q in hachem}}~\label{app: proof Q}
In this appendix, we establish element-wise convergence of $\mathbf{Q}^{m}(z)$ using RMT techniques. 
The proof for $\tilde{\mathbf{Q}}^{m}(z)$ proceeds analogously.
Without loss of generality, we consider the diagonal entry $[\mathbf{Q}^{m}(z)]_{11}$.
Partition the channel matrix as:
$
\mathbf{G}^{m} = \begin{bmatrix}
\mathbf{g}_1^{m} \\
\mathbf{G}_1^{m}
\end{bmatrix},$
where $\mathbf{g}_1^{m} \in \mathbb{C}^{1 \times N_m}$ is the first row. Through block matrix inversion, we obtain
\begin{align*}
[\mathbf{Q}^{m}(z)]_{11}^{-1} &= -z - z N_m^{-1} \mathbf{g}_1^{m} \left( N_m^{-1} \mathbf{G}_1^{m\dagger} \mathbf{G}_1^{m} - z\mathbf{I}_{N_m} \right)^{-1} \mathbf{g}_1^{m\dagger} \\
&= -z - z N_m^{-1} \mathbf{h}_1^{m} \mathbf{A}_1^m \mathbf{T}(z) \mathbf{A}_1^m \mathbf{h}_1^{m\dagger},
\end{align*}
where $\mathbf{T}(z) = ( N_m^{-1} \mathbf{G}_1^{m\dagger} \mathbf{G}_1^{m} - z\mathbf{I}_{N_m} )^{-1}$ and $\mathbf{A}_1^m$ is a diagonal matrix with elements $[l_{1,1}^{m},\dots,{l}_{1,N_m}^{m}]^\dagger$ along its main diagonal.
The normalized channel vector $N_m^{-1/2}\mathbf{h}_1^{m}$ satisfies $\mathbb{E}[|N_m^{-1/2}h_{1,j}^{m}|^8] = O(N_m^{-4})$, and $ \mathbf{A}_1^m\mathbf{T}(z) \mathbf{A}_1^m$ has uniformly bounded spectral norm. 
Using the trace lemma~\cite[Theorem 3.4]{Couillet_Debbah_2011}, we have
\vspace{-0.5em}
\begin{align*}
\lim_{N_m \to \infty, N_m/K_m = \beta_m} N_m^{-1} &\mathbf{h}_1^{m} \mathbf{A}_1^m \mathbf{T}(z) \mathbf{A}_1^m \mathbf{h}_1^{m\dagger} \\
&- N_m^{-1} \operatorname{tr} \left( (\mathbf{A}_1^{m})^2 \mathbf{T}(z) \right) = 0.
\end{align*}
\vspace{-0.5em}
By~\cite[Theorem 3.9]{Couillet_Debbah_2011} and \cite[Lemma 2.4]{hachem2008clt}, we have
\begin{align*}
\lim_{N_m \to \infty, N_m/K_m = \beta_m} [\mathbf{Q}^{m}(z)]_{11}^{-1} - \phi^{-1}_{m,1}(z) = 0.
\end{align*}

\section{The proof of Lemma~\ref{lemma: Q2}}\label{app: proof Q2}
By differentiating equations \eqref{iteration phi} with respect to $z$, we get
\begin{align*}
    &z(1 + N_m^{-1}\sum_l \theta_{i,l}^m \psi_{m,l}(z))\frac{d \phi_{m,i}(z)}{d z}+ \phi_{m,i}(z)(1 \\
    &+N_m^{-1}\sum_l \theta_{i,l}^m \psi_{m,l}(z) + N_m^{-1}z \sum_l \theta_{i,l}^m \frac{d \psi_{m,l}(z)}{dz})=0.
\end{align*}
Substituting the coefficients $z(1 + N_m^{-1}\sum_j \theta_{i,j}^m \psi_{m,j}(z))$ with $-\phi_{m,i}^{-1}(z)$ using \eqref{iteration phi} yields
\begin{equation}\label{eq: derive phi}
    \frac{d \phi_{m,i}(z)}{d z} = \phi_{m,i}^2 [1 + N_m^{-1}\sum_l \theta_{i,l}^m (\psi_{m,l}(z) +z\frac{d \psi_{m,l}(z)}{d z}) ].
\end{equation}
Note that the derivative of the resolvent matrix $\mathbf{Q}^m(z)$ and $\mathbf{\tilde Q}^{m}(z)$ with respect to $z$ is
\begin{align*}
    \frac{d \mathbf{Q}^m(z)}{d z} &= \mathbf{Q}^m(z) \cdot \mathbf{Q}^m(z) = \mathbf{S}^m(z).
\end{align*}
Applying the result of Lemma~\ref{lemma:Q in hachem},
we have 
$\boldsymbol{\lambda}_m(z)=\frac{d \boldsymbol{\phi}_m(z)}{d z}.$
Inserting this into \eqref{eq: derive phi}
yields the fixed-point equation for $\lambda_{m,k}(z)$.
The corresponding equation for $\mu_{m,j}(z)$ follows
immediately by symmetry and is omitted for brevity.

\end{document}